\documentclass[a4paper,USenglish,cleveref, autoref, numberwithinsect]{lipics-v2019}

\hideLIPIcs

\title{Coreset-based Strategies for Robust Center-type Problems}

\titlerunning{Coreset-based Strategies for Robust Center-type Problems}

\author{Andrea Pietracaprina}{Department of Information Engineering, University of Padova, Padova, Italy}{andrea.pietracaprina@unipd.it}{}{}
\author{Geppino Pucci}{Department of Information Engineering, University of Padova, Padova, Italy}{geppino.pucci@unipd.it}{}{}
\author{Federico Sold\`a}{Department of Information Engineering, University of Padova, Padova, Italy}{federico.solda@studenti.unipd.it}{}{}

\authorrunning{A. Pietracaprina, G. Pucci, and F. Sold\`a}

\Copyright{Andrea Pietracaprina, Geppino Pucci, and Federico Sold\`a}



\nolinenumbers 

\usepackage[linesnumbered,ruled]{algorithm2e}

\SetKwRepeat{Do}{do}{while} 

\begin{document}
\maketitle

\begin{abstract}
Given a dataset $V$ of points from some metric space, the popular
$k$-center problem requires to identify a subset of $k$ points
(centers) in $V$ minimizing the maximum distance of any point of $V$
from its closest center.  The \emph{robust} formulation of the problem
features a further parameter $z$ and allows up to $z$ points of $V$
(outliers) to be disregarded when computing the maximum distance from
the centers. In this paper, we focus on two important constrained
variants of the robust $k$-center problem, namely, the \emph{Robust Matroid
Center} (RMC) problem, where the set of returned centers are
constrained to be an independent set of a matroid of rank $k$ built on
$V$, and the \emph{Robust Knapsack Center} (RKC) problem, where each element
$i\in V$ is given a positive weight $w_i<1$ and the aggregate weight
of the returned centers must be at most 1. We devise coreset-based
strategies for the two problems which yield efficient sequential,
MapReduce, and Streaming algorithms. More specifically, for any fixed
$\epsilon>0$, the algorithms return solutions featuring a
$(3+\epsilon)$-approximation ratio, which is a mere additive term
$\epsilon$ away from the 3-approximations achievable by the best known
polynomial-time sequential algorithms for the two problems.  Moreover,
the algorithms obliviously adapt to the intrinsic complexity of the
dataset, captured by its doubling dimension $D$.
For wide ranges of the parameters $k,z,\epsilon, D$, we obtain a
sequential algorithm with running time linear in $|V|$, and
MapReduce/Streaming algorithms with few rounds/passes and
substantially sublinear local/working memory.
\end{abstract}

\section{Introduction}
Center-based clustering is a crucial primitive for data management,
with application domains as diverse as recommendation systems,
facility location, database search, bioinformatics, content
distribution systems, and many more \cite{HennigMMR15}. In general
terms, given a dataset $V$, a distance function between pairs of
points in $V$, and a value $k$, a solution for center-based clustering
is a set of $k$ representative points, called \emph{centers}, which
induce a partition of $V$ into $k$ subsets (clusters), each containing
all points in $V$ closest to the same center. One important
formulation of center-based clustering is the \emph{$k$-center}
problem, where the set of centers must be chosen as a subset of $V$
which minimizes the maximum distance of any point of $V$ to its
closest center. It is well known that $k$-center is $NP$-hard, that it
admits a $2$-approximation algorithm, and that for any $\epsilon> 0$
it is not $(2-\epsilon)$-approximable unless $P=NP$ \cite{Gonzalez85}.

A number of natural variants of $k$-center have been studied in the
literature.  The constrained variants introduced in \cite{ChenLLW16}
restrict the set of returned centers to obey an additional constraint,
which can be expressed either as a \emph{matroid constraint}, that is,
the set of centers must be an independent set of a given matroid
system $(V,I)$ defined on the input dataset $V$, or a \emph{knapsack
  constraint}, where each point in $V$ carries a weight, and the
aggregate weight of the returned centers cannot exceed a certain
budget. Matroid and knapsack constraints arise naturally in the
context of recommendation systems or facility location. In the former
context, consider for instance the case of points in the dataset
belonging to different categories, where all categories should have a
given quota of representatives (centers) in the returned solution, a
constraint naturally expressible as a partition matroid. In the
latter, "opening" a center at a given location might carry different
costs, and the final solution cannot exceed a total budget.

Another variant of the original problem is motivated by the
observation that the $k$-center objective function involves a maximum,
thus the optimal solution is at risk of being severely influenced by a
few “distant” points in $V$, called \emph{outliers}. In fact, the
presence of outliers is inherent in many datasets, since these points
are often due to artifacts or errors in data collection. To cope with
this issue, $k$-center admits the following \emph{robust} formulation
that takes into account outliers \cite{Charikar2001}: given an
additional input parameter $z$, when computing the $k$-center
objective function, the $z$ points of $V$  with the largest distances from
their respective centers are disregarded in the computation of the
maximum.  Robust formulations of the constrained variants have
been also studied, referred to as \emph{Robust Matroid Center} (RMC) and 
\emph{Robust Knapsack Center} (RKC) problems, respectively \cite{ChenLLW16}.

The explosive growth of data that needs to be processed in modern
computing scenarios often rules out the use of traditional sequential
strategies which, while efficient on small-sized datasets, often prove
to be prohibitive on massive ones. It is thus of paramount importance
to devise clustering strategies amenable to the typical computational
frameworks employed for big data processing, such as MapReduce and
Streaming \cite{LeskovecRU14}. \emph{Coreset-based} strategies have
recently emerged as ideal approaches for big data
processing. Informally, these strategies entail the (efficient)
extraction of a very succinct summary $T$ (dubbed \emph{coreset}) of
the dataset $V$, so that a solution for $V$ can be obtained by running
(suitable modifications of) the best sequential algorithm on $T$.
Coreset constructions that can be either parallelized or streamlined
efficiently yield scalable and space-efficient algorithms in the big
data realm.  To objective of this paper is to devise novel
coreset-based strategies for the RMC and RKC problems, featuring
efficient sequential, MapReduce and Streaming implementations.

\subsection{Previous work}

Due to space constraints we only report on the works most closely
related to the specific topic of this paper, and refer the interested
reader to \cite{AwasthiB15} and references therein for a more
comprehensive overview on center-based clustering.  Sequential
approximation algorithms for the RMC and RKC problem are given in
\cite{ChenLLW16,HarrisPST19, ChakrabartyN19}. The best algorithms to
date are sequential 3-approximations for both RMC \cite{HarrisPST19,
  ChakrabartyN19} and RKC \cite{ChakrabartyN19}. All of these
algorithms, however, do not seem immediately amenable to MapReduce or
Streaming implementations. Coreset-based Streaming algorithms for RMC
and RKC have been recently devised by Kale in \cite{Kale19}. For
$\epsilon >0$, Kale's streaming algorithms compute a coreset of size
$O(k(k+z)\log(1/\epsilon)/\epsilon)$ containing a
$(15+\epsilon)$-approximate solution, where $z$ is the number of
outliers and $k$ is the rank of the matroid, for RMC, or the
\emph{maximum} cardinality of a feasible solution, for the RKC
problem. The solution embedded in the coresets of \cite{Kale19} can
be extracted using a brute-force approach. Alternatively, one of the
3-approximate sequential algorithms in \cite{HarrisPST19,
  ChakrabartyN19} can be run on the coreset to yield a
$(51+\epsilon)$-approximate solution.  To the best of our knowledge no
MapReduce algorithms for RKC and RKC have been presented in the open literature.

Coreset-based algorithms for the MapReduce and Streaming setting for
the unconstrained (robust) $k$-center problem and related problems can
be found in \cite{MalkomesKCWM15,CeccarelloPPU17,CeccarelloPP19}. Useful
techniques to deal with matroid constraints in big data scenarios have
been introduced in \cite{CeccarelloPP18,CeccarelloPP20} in the realm
of diversity maximization.

\subsection{Our contribution}
By leveraging ideas introduced in \cite{CeccarelloPP18,Kale19}, we
present novel algorithms for the RMC and RKC problems which attain
approximation ratios close to the best attainable ones, and feature
efficient sequential implementations as well as efficient
implementations in the MapReduce and Streaming settings, thus proving
suitable for dealing with massive inputs. Our strategies exploit the
basic $k$-center primitive to extract a small coreset $T$ from the
input set $V$, with the property that the distance between each point of $V$ and
the closest point of $T$ is a small fraction of cost of the
optimal solution. Also, $T$ contains a good solution for the original
problem on $V$ which can be computed by assigning a suitable
multiplicity to each point of $T$ and by running the best-known
sequential algorithms for RMC and RKC on $T$, adapted to take
multiplicities into account.

More specifically, for any fixed $\epsilon \in (0,1)$, our RMC and RKC
algorithms feature a $(3+\epsilon)$ approximation ratio (see
Corollaries~\ref{cor:mat} and \ref{cor:knap} for a formal statement of
the results). Let $z$ be the number of outliers and let $k$ denote the
matroid rank, in the RMC problem, and the \emph{minimum} cardinality
of an optimal solution, in the RKC problem.  The time and space
requirements of the algorithms are analyzed in terms of $z$, $k$, the
approximation quality, captured by $\epsilon$, and the \emph{doubling
  dimension} $D$ of the input set $V$, a parameter that generalizes
the notion of Euclidean dimension to arbitrary metric spaces.  We
remark that this kind of dimensionality-aware analysis is particularly
relevant in the realm of big data, and it has been employed in a
variety of contexts including diversity maximization, clustering,
nearest neighbour search, routing, and machine learning (see
\cite{CeccarelloPPU17} and references therein).

For both problems, the sequential complexity of our algorithms is
$O(|V| f(z,k,\epsilon,D))$, for a certain function
$f(z,k,\epsilon,D)$, and it is thus linear for fixed values of
$k,z,\epsilon$ and $D$.  The RMC strategy admits a 2-round MapReduce
implementation requiring local memory sublinear in $|V|$
(Theorem~\ref{impl:MRmat}), and a 1-pass Streaming implementation with
working memory size dependent only on $z,k,\epsilon$ and $D$
(Theorem~\ref{impl:STRmat}).  The RKC strategy admits an $R$-round
MapReduce implementation requiring local memory sublinear in $|V|$
(Theorem~\ref{impl:MRknap}), and an $R$-pass Streaming implementation
with working memory size dependent only on $z,k,\epsilon$ and $D$
(Theorem~\ref{impl:STRknap}), where
$R=O(\log(k+z)+D\log(1/\epsilon))$.  For constant $\eta \in (0,1)$,
the number of rounds (resp., passes) $R$ can be reduced to
$O(1/\eta)$, at the expense of a $O(|V|^{\eta/2})$ (resp.,
$O(|V|^{\eta})$) increase in the local memory (resp., working memory)
size. Remarkably, while the
analysis of our algorithms is performed in terms of the doubling
dimension $D$ of $V$, the algorithms are oblivious to the
value $D$ which, in fact, would be difficult to estimate.

Our MapReduce algorithms provide the first efficient solutions to RMC
and RKC in a distributed setting and attain an approximation quality
that can be made arbitrarily close to that of the best sequential
algorithms. Our Streaming algorithms share the same approximation
quality as the MapReduce algorithms and substantially improve upon the
approximations attained in \cite{Kale19}.  Furthermore, all of our
algorithms are very space efficient for a wide range of the parameter
space. In particular, the working space of our RKC Streaming algorithm
depends on the size of the smallest optimal solution rather than on
the largest feasible solution as in \cite{Kale19}, which might result
in a considerable space-saving. Finally, it is important to observe that
in the sequential and Streaming settings, for fixed values of 
$k$, $z$ and $D$, exhaustive search on the coresets yields 
$(1+\epsilon)$-approximate solutions to RMC and RKC 
with work merely linear in $V$.

The rest of the paper is organized as
follows. Section~\ref{sec:prelim} introduces some key technical
notions and formally defines the RMC and RKC problems. The
coreset-based strategies for RMC and RKC are described and analyzed
in Sections~\ref{sec:mat} and \ref{sec:knap}, respectively, while
their MapReduce and Streaming implementations are discussed in
Section~\ref{sec:implementation}. Section~\ref{sec:conclusions} offers
some concluding remarks. 

\section{Preliminaries} \label{sec:prelim}
This section introduces some key notions and basic properties that
will be used throughout the paper, and defines the computational
problems studied in this work.
\subsection{Matroids}

Let $V$ be a ground set of elements from a metric space with distance
function $d(\cdot,\cdot)$ satisfying the triangle inequality. A
\emph{matroid}~\cite{Oxley06} on $V$ is a pair $M = (V,I)$, where $I$
is a family of subsets of $V$, called \emph{independent sets},
satisfying the following properties: (i) the empty set is independent;
(ii) every subset of an independent set is independent
(\emph{hereditary property}); and (iii) if $A \in I$ and $B \in I$,
and $|A| > |B|$, then there exist $x \in A\setminus B$ such that $B
\cup \{x\} \in I$ (\emph{augmentation property}). 
An independent set
is \emph{maximal} if it is not properly contained in another
independent set. A basic property of a
matroid $M$ is that all of its maximal independent sets have the same
size. The notion of maximality can be naturally extended to any subset of the
ground set. Namely, for $V' \subseteq V$,
an independent set $A \subseteq V'$ of
maximum cardinality among all independent sets contained in $V'$
is called a maximal independent set of $V'$, and all maximal independent sets of $V'$ have the same size. We  let the \emph{rank} 
 of a subset $V' \subset V$, denoted by $rank(V')$ to be the size of
a maximal independent set in $V'$.  The rank of the matroid $rank(M)$ is then defined as   
$rank(V)$.   An important property of the rank function is \emph{submodularity}: for any 
$A,B \subseteq V$ it holds that $rank(A\cup B) + rank(A\cap B)\leq rank(A)+rank(B)$. The following lemma is an adaptation of \cite[Lemma 3]{Kale19} and provides a 
useful property of matroids which will be  exploited to
derive the results of this paper.
\begin{lemma}[Extended augmentation property]\label{lemma:extAug}
Let $M=(V,I)$ be a matroid. Consider an independent set $A \in
I$, a subset $V' \subseteq V$, and a maximal independent set $B$ of
$V'$.  If there exists $y \in V'\setminus A$ such that $A \cup
\{y\} \in I$, then there exists $x \in B \setminus A$ such that $A \cup
\{x\} \in I$.
\end{lemma}
\begin{proof}
Since $B$ is maximal in in $V'$, we have that
$rank(B\cup\{y\}) = rank(B)= rank((B\cup\{y\}) \cap (A\cup B))$. Also, 
$rank((B\cup\{y\}) \cup (A\cup B)) \geq rank(A\cup\{y\})  \geq  |A|+1$, since $A\cup \{y\}\in I$.
By applying the submodularity property to sets $B\cup \{y\}$ and $A\cup B$ we have the inequality
$$rank((B\cup\{y\}) \cup (A\cup B)) + rank((B\cup\{y\}) \cap (A\cup B)) \leq rank(B\cup\{y\}) + rank(A\cup B),$$
which can be manipulated using the above relations  to yield $rank(A\cup \{y\})  \leq rank(A\cup B)$, whence  $rank(A\cup B) \geq |A|+1$.
So, there exists an independent set $C\subseteq A\cup B$ of $|A|+1$ elements, and the lemma follows. 
\end{proof}

\subsection{Definitions of the problems}
The well-known
\emph{$k$-center problem} is defined as follows. Given a set $V$ of
points from a metric space with distance function $d(\cdot,\cdot)$,
determine a subset $S \subseteq V$ of size $k$ which minimizes
$\max_{i \in V} \min_{c \in S} d(i,c)$. For convenience, throughout
the paper we will use the notation $d(i,S)=\min_{c \in S}
d(i,c)$. Several variants of the $k$-center problem have been proposed
and studied in the literature. Mostly, these variants impose
additional constraints on the solution $S$ and/or allow a given number
of points to be disregarded from the computation of the maximum in the
objective function. In this paper, we focus on two of these variants
defined below using the same terminology adopted in
\cite{HarrisPST19}.

\begin{definition} \label{matroid}
Let $M=(V,I)$ be a matroid
defined over the set of points $V$, 
and let $z$ be an integer, with $0 \leq z < |V|$.
The \textbf{Robust Matroid Center $($RMC$)$ problem} on $M$ with
parameter $z$, requires 
to determine a set $S \in I$ minimizing  
\[
r(S,V,z)=
\min_{X \subseteq V: |X|\geq |V|-z} \max_{i \in X} d(i,S).
\]
We use the tuple $(M=(V,I),z)$ to denote an instance of RMC.
\end{definition}
\begin{definition} \label{knapsack}
Let $V$ be a set of points. Suppose that for each $j \in V$ a 
weight $w_j \in [0,1]$ is given and 
let $z$ be an integer, with $0 \leq z < |V|$.
The \textbf{Robust Knapsack Center $($RKC$)$ problem} 
on $V$ with parameter $z$ and weights $\mathbf{w}=\{w_i : i \in V\}$, 
requires 
to determine a set $S \subseteq V$ with $\sum_{j \in S} w_j \leq 1$,
minimizing
\[
r(S,V,z)=
\min_{X \subseteq V: |X|\geq |V|-z} \max_{i \in X} d(i,S).
\]
We use the tuple $(V,z,\mathbf{w})$ to denote an instance of RKC.
\end{definition}

The RMC and RKC problems share the same cost function
$r(S,V,z)$ but exhibit different feasible solutions for the same ground set $V$.
Observe that $r(S,V,z)$ coincides with the $(|V|-z)$-th
smallest distance of a point of $V$ from $S$. In other words, the best
solution is allowed to ignore the contribution of the $z$ most distant
points, which can be regarded as outliers. 

The state of the art on sequential approximation algorithms for the two problems are the 3-approximation algorithms
for the RMC and RKC problems presented in \cite{ChakrabartyN19}.  
The coreset-based approaches developed in this paper require
the solution of generalized versions of the above two problems, where
each point $i \in V$ comes with a positive integer multiplicity $m_i$.
Let $\mu_V = \sum_{i  \in V} m_i$. The generalized versions of the two problems,
dubbed {\bf RMC problem with Multiplicities (RMCM problem)} and
{\bf RKC problem with Multiplicities (RKCM problem)}, respectively,
allow $z$ to vary in $[0, \mu_V)$ and modify the cost function 
as follows:
\[
r(S,V,z)=
\min_{X \subseteq V: \sum_{i \in X} m_i \geq \mu_V-z} \max_{i \in X} d(i,S).
\]
Letting $\mathbf{m}=\{m_i : i \in V\}$, we use the tuples
$(M=(V,I),z,\mathbf{m})$ and $(V,z,\mathbf{w},\mathbf{m})$ to denote
instances of RMCM and RKCM, respectively.  To the best of our
knowledge, prior to this work, no algorithms had been devised to solve
the RMCM and RKCM problems. However, in the rest of the subsection we describe how the 
sequential algorithms in
\cite{ChakrabartyN19} can be easily adapted to solve the more general
RMCM and RKCM problems, featuring the same 3-approximation guarantee as
in the case without multiplicities.

 We start by giving the definition of \emph{Robust
  $\mathscr{F}$-Supplier problem with Multiplicities}, which
generalizes the Robust $\mathscr{F}$-Supplier problem of
\cite{ChakrabartyN19}, and recall the definition of the auxiliary
\emph{$\mathscr{F}$-maximization under Partition Constraint}
($\mathscr{F}$-PCM) problem \cite{ChakrabartyN19}.

An instance of the Robust $\mathscr{F}$-Supplier problem with
Multiplicities is a tuple $\mathscr{I}= (F,C,d,m,\mathscr{F},\mu)$
where $(F\cup C, d)$ is a metric space, $m$ is an integer parameter,
$\mathscr{F}\subseteq 2^F$ is a down-closed family of subsets of $F$,
and $\mu:C\rightarrow \mathbb{N}$ is a function that associates to each point of $C$ its multiplicity . The objective is to find $S\in
\mathscr{F}$ and $T\subseteq C$ for which $\sum_{u\in T} \mu(u)\geq m$
and $\max_{u\in T}d(u,S)$ is minimized. An instance
of $\mathscr{F}$-PCM is a tuple $\mathscr{I}=(F,\mathscr{F},
\mathscr{P},val)$, where $F$ is a finite set, $\mathscr{F}\subseteq
2^F$ is a down-closed family of subsets of $F$, $\mathscr{P}$ is a
sub-partition of $F$, and $val:F\rightarrow \{0,1,2,\dots\}$ is
integer valued function consistent with 
$\mathscr{P}$ in the sense that for each $A \in \mathscr{P}$
and for each pair $f_1, f_2, \in A$,
$val(f_1) = val(f_2)$. For a set $S \in \mathscr{P}$, we let
$val(S) = \sum_{f \in S} val(f)$. 
The objective of $\mathscr{F}$-PCM problem is to compute
\[
\max \{val(S) : S \in \mathscr{F} \wedge |S\cap A| \leq 1 
\forall A \in \mathscr{P}\}.
\]
The following theorem extends
\cite[Theorem 1]{ChakrabartyN19} to encompass multiplicities.

\begin{theorem}\label{theorem:A1}
Let $\mathcal{A}$ be an algorithm for the $\mathscr{F}$-PCM problem,
and let $T_{\mathcal{A}}(\cdot)$ denote its complexity.  Given an
instance $\mathscr{I}= (F,C,d,m,\mathscr{F},\mu)$ of the Robust
$\mathscr{F}$-Supplier problem with Multiplicities, consider the
instance $\mathscr{I}'=(F,\mathscr{F},\mathscr{P},val)$ of
$\mathscr{F}$-PCM. Then, there is an
algorithm for the Robust $\mathscr{F}$-Supplier
problem with Multiplicities which returns a $3$-approximate solution
to $\mathscr{I}$ in time
$\mbox{poly}(|\mathscr{I}|)T_{\mathcal{A}}(\mathscr{I}')$.
\end{theorem}

\begin{proof}
The proof of this theorem, follows the same reasoning of \cite[Theorem
  1]{ChakrabartyN19}, hence we describe here only the differences with
respect to that proof.

In Algorithm 1, described in \cite[Section 3.1]{ChakrabartyN19}, we
substitute Line 10 with the following line: 
\[
val(f)\leftarrow \sum_{u\in Chld(v)} \mu(u)\ \forall f \in B_F(v,1).
\]
Next, we
substitute the politope $\mathscr{P^I}_{cov}$ defined at the beginning
of \cite[Section 3.2]{ChakrabartyN19}, with the one described by the
constraints below. (Note that only the first constraint is different
with respect to the original ones.)
\begin{equation}
\sum_{v\in C}m(v) cov(v) \geq m  
\tag{$\mathscr{P^I}_{cov}$.1'}
\end{equation}
\begin{equation}
cov(v)- \sum_{S\in \mathscr{F}:d(v,S)\leq 1} z_S =0, \quad \forall v\in C
\tag{$\mathscr{P^I}_{cov}$.2}
\end{equation}
\begin{equation}
\sum_{S\in \mathscr{F}} z_S = 1
\tag{$\mathscr{P^I}_{cov}$.3}
\end{equation}
\begin{equation}
z_S \geq 0, \quad \forall S \in \mathscr{F}
\tag{$\mathscr{P^I}_{cov}$.4}
\end{equation}

The remaining part of the proof, follows exactly the same passages
as the original proof. However, we need the following modified
version of \cite[Claim 7]{ChakrabartyN19}, whose proof requires
only straightforward adaptations to accommodate multiplicities.

\begin{claim}[Modified Claim 7 of \protect{\cite{ChakrabartyN19}}]
Let $S\in \mathscr{F}$ be any feasible solution of the
$\mathscr{F}$-PCM instance constructed by Algorithm 1. Then, \[
\sum_{f\in S} val(f) = \sum_{v\in R(S)} \sum_{u\in Chld(v)} \mu(u).\]
\end{claim}
\end{proof}

Since  the RMCM and RKCM problem can be seen as instantiations of the
Robust $\mathscr{F}$-Supplier problem with Multiplicities,
Theorem \ref{theorem:A1}, combined with the $\mathscr{F}$-PCM algorithm 
from \cite{ChakrabartyN19}, allows us to derive the result stated in the following theorem.
\begin{theorem} \label{thm:approxmult}
There exist $3$-approximate polynomial-time sequential algorithms 
for the RMCM and RKCM problem.
\end{theorem}

\subsection{Doubling dimension}
The algorithms in this paper will be analyzed in terms of the
dimensionality of the ground set $V$ as captured by the
well-established notion of doubling dimension. Formally, given a point
$i \in V$, let the \emph{ball of radius $r$ centered at $i$} be the subset of
points of $V$ at distance at most $r$ from $i$. The \emph{doubling
  dimension} of $V$ is the smallest value $D$ such that any
balls of radius $r$ centered at a point $i \in V$ 
is contained in the union of at most
$2^D$ balls of radius $r/2$ suitably centered at points of $V$.
The algorithms that will be presented in this paper adapt
automatically to the doubling dimension $D$ of the input dataset and
attain their best performance when $D$ is small, possibly constant.
This is the case, for instance, of ground sets $V$ whose points belong to 
low-dimensional Euclidean spaces, or represent nodes of
mildly-expanding network topologies under shortest-path distances.

The doubling dimension $D$ of a ground set $V$ allows the following
interesting characterization of how the radius of a $k$-center
clustering decreases as $k$ increases, which will be crucially
exploited in this paper.

\begin{proposition}\label{prop:dublingdim}
Let $\epsilon \in (0,1)$. 
Consider a set $S \subseteq V$ of size $k$, and let $r = \max_{i \in V} d(i,S)$.
If $V$ has doubling dimension $D$, there exists a set $S' \subseteq V$
of size $\leq k (1/\epsilon)^D$ such that
$\max_{i \in V} d(i,S') \leq \epsilon r$.
\end{proposition}
\begin{proof}
By repeatedly applying the
definition of doubling dimension, it is easily seen that each ball
of radius $r$ around a point in $S$ can be covered with at most
$(1/\epsilon)^{D}$ smaller balls of radius $\epsilon r$. The centers
of all of these smaller balls provide the desired set $S'$. 
\end{proof}

\section{Coreset-based strategy for the RMC problem}\label{sec:mat}

In this section, we present a two-phase strategy to solve the RMC
problem based on the following simple high-level idea. In the first
phase we extract a small \emph{coreset} $T$ from the ground set $V$,
that is, a subset of $V$ with the property that each point $j \in V$
has a suitably ``close'' \emph{proxy} $p(j)$ in $T$. In the second
phase, an approximate solution $S$ to the RMCM problem on $T$ is
computed, where the multiplicity $m_i$ of each $i \in T$ is defined as
the number of distinct points $j \in V$ whose proxy is $i$.  In what
follows, we first determine sufficient conditions on the coreset $T$
which guarantee that a good solution to the RMCM problem on $T$ is
also a good solution for the RMC problem on $V$, and then describe how
such a coreset can be constructed, analyzing its size in terms of the
doubling dimension of $V$.

Let $(M=(V,I),z)$ be an instance of the RMC problem and 
$r^*(M,z)$ be the cost of its optimal solution. Consider a coreset $T
\subseteq V$ with proxy function $p : V \rightarrow T$, and let $m_i
= |\{j \in V : p(j)=i\}|$, for every $i \in T$. 
Let $M_T=(T,I_T)$ denote the restriction of matroid $M=(V,I)$ to the
coreset $T$, where for each $X \in I$, $X \cap T \in I_T$.
Finally, let $(M_T,z,\mathbf{m})$ denote 
the RMCM instance defined by $M_T$,
$z$ and $\mathbf{m}=\{m_i : i \in T\}$. We have:

\begin{lemma}\label{lem:mat_one}
Let $\epsilon' \in (0,1)$ be a design parameter. Suppose that the
coreset $T$ with proxy function $p : V \rightarrow T$ satisfies the
following conditions:
\begin{description} 
\item[C1]
For each $j \in V$, $d(j,p(j)) \leq \epsilon' r^*(M,z)$;
\item[C2]
For each independent set $X \in I$ there exists an injective mapping
$\pi_X : X \rightarrow T$ such that:
\begin{itemize}
\item
$\{\pi_X(i) : i \in X\} \subseteq T$ is an independent set;
\item
for each $i \in X$, $d(i,\pi_X(i)) \leq \epsilon' r^*(M,z)$.
\end{itemize}
\end{description} 
Then: 
\begin{description} 
\item[P1]
There exists a solution to $(M_T=(T,I_T),z,\mathbf{m})$  of
cost at most $(1+2\epsilon') r^*(M,z)$;
\item[P2]
Every solution $S$ to $(M_T,z,\mathbf{m})$ of cost
$r_S$ is also a solution to $(M=(V,I),z)$ of cost
$r_S+\epsilon'r^*(M,z)$.
\end{description} 
\end{lemma}
\begin{proof}
Let us first show P1. 
Let $X^*_V$ be the optimal solution to the RMC instance $(M=(V,I),z)$
and let $Y = \{\pi_{X^*_V}(o) : o \in X^*_V \} \subseteq T$.
We will show that $Y$ is a $(M_T,z,\mathbf{m})$  of
cost at most $(1+2\epsilon') r^*(M,z)$.
By C2, $|Y| = |X^*_V|$ and $Y$ is an independent set in $I_T$.
Consider now a point $j\in V$ such that 
$\exists o \in X^*_V$ with $d(j,o) \leq r^*(M,z)$ and observe that there
are at least $|V|-z$ such points (e.g., all nonoutliers).
We have that
\begin{eqnarray*}
d(p(j),Y) & \leq & d(p(j),\pi_{X^*_V}(o)) \\
& \leq & d(p(j),j)+d(j,o) + d(o,\pi_{X^*_V}(o)) 
\;\;\; \mbox{ (by triangle inequality)} \\
& \leq & \epsilon'r^*(M,z) + r^*(M,z)+\epsilon'r^*(M,z) 
\;\;\; \mbox{ (by C1 and C2)} \\
& \leq & (1+2 \epsilon')r^*(M,z).
\end{eqnarray*}
Let $\mu_T = \sum_{i \in T} m_i$ and observe that $\mu_T=|V|$.
We have that
\begin{eqnarray*}
\sum_{i \in T: d(i,Y) \leq (1+2 \epsilon')r^*(M,z)} m_i 
&\geq & \sum_{i \in T: \exists j \in V:
  (i=p(j)) \wedge (d(j,X^*_V) \leq r^*(M,z))} m_i \\
&\geq & |\{j \in V :  d(j,X^*_V) \leq r^*(M,z)\}| \\
&\geq & |V|-z,
\end{eqnarray*}
which concludes the proof of P1. In order to prove P2, let $S$
be a solution to $(M_T,z,\mathbf{m})$ of cost
$r_S$. Clearly, $S$ is an independent set in $I$.
Consider a generic point $i \in T$ such that 
$d(i,S) \leq r_S$ and let $a$ be the point of $S$ closest
to $i$. Observe that the $m_i$ points $j \in V$ 
with $i=p(j)$ are such that 
$d(j,S) \leq d(j,a) \leq d(j,i)+d(i,a) \leq \epsilon' r^*(M,z) +r_S$.
Since $\sum_{i\in T: d(i,S)\leq r_S} m_i \geq \mu_T-z$, 
there are at least  $\mu_T-z = |V|-z$ points of $V$
that are within a distance $\epsilon' r^*(M,z)+r_S$ from $S$.
\end{proof}
Later in this section (see Theorem~\ref{thm:mat_two}) we will show
that if coreset $T$ exhibits properties P1 and P2 stated in the above
lemma, then a good solution to the $(M,z)$ RMC instance can be
obtained by running an approximation algorithm for RMCM on $T$.  We
now show how to construct a coreset $T$ satisfying Conditions C1 and
C2 of Lemma~\ref{lem:mat_one} (hence, exhibiting properties P1 and P2
by virtue of the lemma).  The construction strategy is simple and, as
will be discussed in Section~\ref{sec:implementation}, also features
efficient MapReduce and Streaming implementations.  As in previous
works, we assume that constant-time oracles are available to compute
the distance between two elements of $V$ and to check the independence
of a subset of $V$ (see e.g., \cite{AbbassiMT13}).  Let $k$ be the
rank of matroid $M$. We make the reasonable assumption that $k$ is
known to the algorithm. Also, for ease of presentation, we restrict the
attention to matroids $(V,I)$ such that $\{j\} \in I$ for every $j \in
V$. This restriction can be easily removed with simple modifications
to the algorithms.

In order to construct the coreset $T$, we first compute a
$\beta$-approximate solution $T_{k+z}$ to $(k+z)$-center on $V$ and
determine $r_{T_{k+z}} = \max_{j \in V} d(j,T_{k+z})$. In the sequential setting,
Gonzalez's algorithm \cite{Gonzalez85}, provides a $\beta=2$ approximation\footnote{In the streaming setting, Gonzalez's algorithm cannot be used, and a slightly larger value of $\beta$ will be needed} and computes $T_{k+z}$ and
$r_{T_{k+z}}$ in $O((k+z)|V|)$ time. Then, we  compute a set $T_{\tau}$ of $\tau$ points of $V$
such that $d(i_1,i_2) > (\epsilon'/(2\beta)) r_{T_{k+z}}$, for every $i_1 \neq
i_2 \in T_{\tau}$, and $d(j,T_{\tau}) \leq  (\epsilon'/(2\beta)) r_{T_{k+z}}$, for
every $j \in V$. Hence, $r_{T_{\tau}} = \max_{i \in V} d(i,T_{\tau})
\leq  (\epsilon'/(2\beta)) r_{T_{k+z}}$.  Clearly, the value $\tau$ will depend
on $r_{T_{k+z}}$, $\epsilon'$, and $\beta$.  $T_{\tau}$ can be
computed in $O(\tau |V|)$ time by adapting the well known greedy strategy by Hochbaum and Shmoys
 \cite{HochbaumS85}, namely, by performing a linear scan of $V$ and adding to  (an initially empty)
$T_{\tau}$ all those points $j \in V$ at distance greater
than $(\epsilon'/(2\beta)) r_{T_{k+z}}$ from the current $T_{\tau}$.  (Observe
that $\tau$ is the size of the final set $T_{\tau}$ but the
construction does not require the knowledge of $\tau$.) Let $T_{\tau}
= \{i_1, i_2, \ldots, i_{\tau}\}$ and, for $1 \leq \ell \leq \tau$,
define the cluster $C_{\ell} = \{j \in V :
d(j,i_{\ell})=d(j,T_{\tau})\}$ (ties broken arbitrarily 
for points $j \in V$ equidistant from two or more points of $T_{\tau}$).  From each $C_{\ell}$ we extract a
maximum independent set $Y_{\ell}$ and define $T=\cup_{1 \leq \ell
  \leq \tau} Y_{\ell}$.  For every $1 \leq \ell \leq \tau$ and every
point $j \in V \cap C_{\ell}$ we set the proxy $p(j) = i \in
Y_{\ell}$, where $d(j,i)=d(j,Y_{\ell})$ (ties broken
arbitrarily).  For each $i \in T$, its multiplicity is set to $m_i =
|\{j \in V: p(j)=i\}|$. We have:

\begin{lemma} \label{lem:c1c2}
The coreset $T$ constructed by the above algorithm 
satisfies Conditions C1 and C2 of Lemma~\ref{lem:mat_one}.
\end{lemma}
\begin{proof}
First, we prove C1. Consider an arbitrary point $j \in V$ and
suppose that $j$ belongs to 
cluster $C_{\ell}$, for some $\ell$, hence $p(j)$ belongs to 
$Y_{\ell} \subseteq C_{\ell}$ and $d(j,p(j)) \leq 2r_{T_{\tau}}$. 
Let $\rho^*(V,k+z)$ be the cost of the optimal solution to the $(k+z)$-center
problem on $V$. Since any solution to the $(M=(V,I),z)$ instance of RMC,
augmented with the outlier points, is a solution to $(k+z)$-center
on $V$, it is easy to see that $\rho^*(V,k+z) \leq r^*(M,z)$.
Now, by using the fact that $T_{k+z}$ is a $\beta$-approximate 
solution to $(k+z)$-center on $V$, we have
\[
2r_{T_{\tau}} \leq 2(\epsilon'/(2\beta)) r_{T_{k+z}} \leq \epsilon' \rho^*(V,k+z)
\leq \epsilon' r^*(M,z),
\]
thus proving C1. As for C2, we reason as follows. Consider an
arbitrary independent set $X \in I$. We now show that there exists an
injective mapping $\pi_X$ which transforms $X$ into an independent set
contained in $T$, and such that, for each $1 \leq \ell \leq
\tau$ and $j \in X \cap C_{\ell}$, $\pi_X(j) \in Y_{\ell} \subseteq
C_{\ell}$ (i.e., $j$ and $\pi_X(j)$ belong to the same cluster $C_{\ell}$) . This
will immediately imply that $d(j,\pi_X(j)) \leq 2r_{T_{\tau}} \leq
\epsilon' r^*(M,z)$.  Let $X = \{x_a : 1 \leq a \leq |X|\}$. We
define the mapping $\pi_X$ incrementally one element at a
time. Suppose that we have fixed the mapping for the first $h$
elements of $X$ and assume, inductively, that $W(h)=\{\pi_X(x_a): 1
\leq a \leq h\} \cup \{x_a: h < a \leq |X|\}$ is an independent set of
size $|X|$ and that $x_a$ and $\pi_X(x_a)$ belong to the same cluster,
for $1 \leq a \leq h$.  Consider now $x_{h+1}$ and suppose that
$x_{h+1} \in C_{\ell}$, for some $\ell$. We distinguish among the
following two cases:
\begin{itemize}
\item
{\bf Case 1.} If $x_{h+1} \in Y_{\ell}$, we set
$\pi_X(x_{h+1})=x_{h+1}$, hence $W(h+1)=W(h)$.
\item 
{\bf Case 2.} If $x_{h+1} \not\in Y_{\ell}$,
we apply the extended augmentation property stated
in Lemma~\ref{lemma:extAug} with $A=W(h) \setminus \{x_{h+1}\}$,
$y = x_{h+1}$, $V' = C_{\ell}$, and $B=Y_{\ell}$
to conclude that there exists a
point $\pi_X(x_{h+1}) \in B\setminus A = Y_{\ell}\setminus (W(h) \setminus
\{x_{h+1}\})$ such that $W(h+1)=(W(h) \setminus \{x_{h+1}\}) \cup
\pi_X(x_{h+1})$ is an independent set.
\end{itemize}
After $|X|$ iterations of the above inductive argument, we have that
the mapping $\pi_X$ is completely specified and exhibits the following
properties: it is inductive, $\{\pi_X(x_a) : 1 \leq a \leq |X|\}$ is
independent, and, for $1 \leq a \leq |X|$, if $x_a \in C_{\ell}$ then
also $\pi_X(x_a) \in C_{\ell}$, hence $d(x_a,\pi_X(x_a)) \leq
\epsilon' r^*(M,z)$. This proves C2.
\end{proof}

The size of coreset $T$ can be conveniently bounded as a function of the doubling dimension of the ground set $V$.
\begin{theorem}\label{thm:RMCcorsize}
If $V$ has doubling dimension $D$, then the coreset $T$ 
obtained with the above construction has size 
$|T| = O(k(k+z) (4 \beta/\epsilon')^D)$.
\end{theorem}
\begin{proof}
Observe that $|T| \leq k \tau$, hence we need to bound $\tau$. 
Consider the first set $T_{k+z}$ of $k+z$ centers computed
by the coreset construction algorithm. 
Proposition~\ref{prop:dublingdim} implies that there exists
a set $T'$ of at most $h=(k+z)(4\beta/\epsilon')^D$ points such that
$\max_{i \in V} d(i,T') \leq (\epsilon'/(4\beta))r_{T_{k+z}}$, hence
$V$ can be covered with $h$ balls of radius
at most $(\epsilon'/(4\beta))r_{T_{k+z}}$. It is easy to see that
in the adaptation of Hochbaum and Shmoys'
strategy \cite{HochbaumS85} described above to construct $T_{\tau}$,
only one point from each such ball can be added to $T_{\tau}$.
Hence, $\tau = |T_{\tau}| \leq h$, and the theorem follows. 
\end{proof}

\begin{theorem}\label{thm:mat_two}
Let $\epsilon \in (0,1)$ and $\alpha \geq 1$. 
Suppose that the coreset $T$ exhibits Properties P1 and P2 of
Lemma~\ref{lem:mat_one}, for $\epsilon'=\epsilon/(2\alpha+1)$.  Then, an
$\alpha$-approximate
solution $S$ to instance $(M_T=(T,I_T),z,\mathbf{m})$ of RMCM is a
$(\alpha+\epsilon)$-approximate solution to instance $(M=(V,I),z)$ of
RMC.
\end{theorem}

\begin{proof}
By Property P1 of Lemma~\ref{lem:mat_one}, we know that the optimal
solution to $(M_T=(T,I_T),z,\mathbf{m})$ has cost at most
$(1+2 \epsilon')r^*(M,z)=(1+2\epsilon/(2\alpha+1))r^*(M,z)$.
Hence, $S$ has cost $r_S \leq (\alpha+2\alpha\epsilon/(2\alpha+1))r^*(M,z)$.
By Property P2 of Lemma~\ref{lem:mat_one}, $S$ is also a solution 
to instance $(M=(V,I),z)$ of RMC with cost
$r_S+\epsilon' r^*(M,z) \leq (\alpha+\epsilon)r^*(M,z)$.
\end{proof}
The following corollary is an immediate consequence of
Theorems~\ref{thm:RMCcorsize}, \ref{thm:mat_two} and \ref{thm:approxmult}.
\begin{corollary} \label{cor:mat}
For any fixed $\epsilon \in (0,1)$, 
the coreset-based strategy for the RMC problem presented above
can be used to compute 
a $(3+\epsilon)$-approximate solution to any instance $(M=(V,I),z)$.
If $V$ has constant doubling dimension, the sequential running time is 
$O(|V| \mbox{\rm poly}(k,z))$.
\end{corollary}

\section{Coreset-based strategy for the RKC problem}\label{sec:knap}

In this section we present a coreset-based strategy for the RKC
problem which is similar in spirit to the one presented in the
previous section for the RMC problem.  Consider an instance
$(V,z,\mathbf{w})$ of RKC, and let $r^*(V,z,\mathbf{w})$ denote the
cost of an optimal solution.
The idea is to extract a coreset $T$ from
a $\tau$-clustering of $V$ by picking one point per cluster so that
$T$ contains a good solution $S$ for $(V,z,\mathbf{w})$, and then to
run an approximation algorithm for the RKCM problem on $T$, using, for
each $i \in T$, the size of its cluster as multiplicity $m_i$.  The
cost penalty introduced by seeking the solution on $T$ rather than on
the entire set $V$ will be limited by ensuring that for each $i \in
V$, the distance $d(i,T)$ is sufficiently small.  The main difficulty
with the above strategy is the choice of a suitable clustering
granularity $\tau$, hence we resort to testing geometrically
increasing guesses for $\tau$. Observe that in this fashion we
generate a sequence of coresets, thus a sequence of RKCM instances
upon which the approximation algorithm has to be run.  A challenge of
this approach is to devise a suitable stopping condition for detecting
a good guess.

More specifically, our coreset-based strategy, dubbed {\sc
  RKnapCenter}, works as follows.  Let ${\cal A}_{\rm RKCM}$ be an $\alpha$-approximation
algorithm for the RKCM problem, and let $\epsilon \in (0,1)$ be a
fixed accuracy parameter. For each value $\tau$ in a geometric
progression, we run a procedure dubbed {\sc CoresetComputeAndTest}, which
first computes a partition of $V$ into $\tau$ clusters $C_1, C_2,
\ldots, C_{\tau}$, induced by a solution to $\tau$-center on $V$, sets
coreset $T$ to contain one point of minimum weight from each cluster,
and finally runs ${\cal A}_{\rm RKCM}$ on the RKCM instance $(T,z,\mathbf{w}_T,\mathbf{m})$,
where  $\mathbf{w}_T$ is the restriction of $\mathbf{w}$ to $T$, and 
$\mathbf{m} = \{ m_i: i\in T\}$, with $m_i$ being the size of the cluster that $i$ belongs to. {\sc CoresetComputeAndTest} returns $S$,  the
solution computed by ${\cal A}_{\rm RKCM}$, $r_1 = \max_{j \in V} d(j,T)$ and
$r_2 = \min_{X \subseteq T: \sum_{i \in X} m_i \geq \mu_T-z} \max_{i\in X} d(i,S)$,
where $\mu_T = \sum_{i \in T} m_i = |V|$. 
If $\alpha(4\alpha+2)r_1 \leq \epsilon (r_2-4\alpha r_1)$,
then the algorithm terminates and returns $S$ as final
solution. (See Algorithm~\ref{alg:knap_one} for the pseudocode.)

\begin{lemma}\label{lemma:knap_one}
\sloppy
For any  $\tau \geq 1$, consider the triplet $(S, r_1, r_2)$
returned by one execution of {\sc CoresetComputeAndTest}$(\tau)$
within  {\sc RKnapCenter$(\epsilon, \alpha)$}. Then:
\begin{enumerate}
\item \label{knap_one:enum1}
$S$
is a solution to the RKC instance $(V,z,\mathbf{w})$
of cost at most $2 r_1+ r_2$.
\item \label{knap_one:enum2}
$r_2 \leq \alpha(r^*(V,z,\mathbf{w})+4r_1)$.
\end{enumerate}
\end{lemma}
\begin{proof}
Let us prove Point~\ref{knap_one:enum1} first.
Since $S$ is a feasible solution to 
$(T,z,\mathbf{w}_T,\mathbf{m})$ and $\mathbf{w}_T$ is the restriction of $\mathbf{w}_T$ 
of the points of $T$,
$S$ is also feasible for $(V,z,\mathbf{w})$. By 
definition of $r_2$, there exists a subset $X \subseteq T$ such that 
$\sum_{i \in X} m_i \geq |V|-z$ and $\max_{i \in X}d(i,S) \leq r_{2} $. 
Consider a point $i \in X$ and suppose that $i \in C_{\ell}$. Then, by the triangle inequality, $\forall j \in C_{\ell}$,
$d(j,S)\leq d(j,i)+d(i,S)\leq 2r_1+ r_2$. Thus, the points in $V$ 
at distance at most $2r_1+ r_2$ from $S$ are
at least $\sum_{\ell : C_{\ell} \cap X \neq \emptyset}|C_{\ell}| = \sum_{i \in X}
m_i \geq |V|-z$.

As for Point~\ref{knap_one:enum2}, let $S^*_V \subseteq V$ be an
optimal solution to $(V,z,\mathbf{w})$ and let $X = \{i \in T: (i \in
C_{\ell}) \wedge (C_{\ell} \cap S^*_V \neq \emptyset)\}$.  We now
show that $X$ is a feasible solution to $(T,z,\mathbf{w}_T,\mathbf{m})$ of cost
at most $r^*(V,z,\mathbf{w})+4r_1$, hence $S$ must have a cost of at
most $\alpha(r^*(V,z,\mathbf{w})+4r_1)$. Observe that since
$T$ contains the points of minimum weight from each cluster,
$\sum_{i \in X} w_i \leq \sum_{j \in S^*_V } w_j \leq 1$. 
Consider a point $j \in V$ such that $d(j,S^*_V) \leq r^*(V,z,\mathbf{w})$.
Clearly there are at least $|V|-z$ such points (e.g., all nonoutliers).
Let $L = \{\ell : \exists j \in C_{\ell} \mbox{ with } 
d(j,S^*_V) \leq r^*(V,z,\mathbf{w})\}$. Hence, $\sum_{i \in T\cap C_{\ell}: 
 \ell \in L} m_i \geq |V|-z$. Consider 
a cluster $C_{\ell}$ with $\ell \in L$ and the
point $i \in T \cap C_{\ell}$. 
Since $\ell \in L$, $C_{\ell}$ contains
a point $j$ with $d(j,S^*_V) \leq r^*(V,z,\mathbf{w})$.
Let $o$ be the point of $S^*_V$ closest to $j$ and suppose that
$o$ belongs to cluster $C_{\ell'}$. Letting $i'$ be the point in
$X \cap C_{\ell'}$, by the triangle inequality 
we have $d(i,X) \leq d(i,i') \leq d(i,j)+d(j,o)+d(o,i') \leq
r^*(V,z,\mathbf{w})+4r_1$.
This immediately implies that 
$\sum_{i \in T : d(i,X) \leq r^*(V,z,\mathbf{w})+4r_1} m_i \geq |V|-z$.
\end{proof}

\begin{algorithm}[t]
\SetAlgoLined
 $ \tau \leftarrow 1$ \\
 \Do{$\frac{\alpha(4\alpha +2) r_1}{r_2- 4\alpha r_1} > \epsilon$}{
  $(S, r_1, r_2) \leftarrow \mbox{\sc CoresetComputeAndTest}(\tau)$ \\
  $ \tau \leftarrow 2 \tau $ \\ 
 }
 \Return{$S$ }
 
\vspace*{0.3cm}
{\bf Procedure} {\sc CoresetComputeAndTest}$(\tau)$\;
 $\{C_1, C_2, \ldots, C_{\tau}\} \leftarrow \tau$-clustering  
 induced by a solution $\Gamma$ to $\tau$-center on $V$ \\
 $r_1 \leftarrow \max_{j \in V} d(j,\Gamma)$ \\
 $T \leftarrow \emptyset$ \\
 \ForEach{cluster $C_{\ell}$}{
  $i \leftarrow arg\min_{j\in C_{\ell}} w_{j}$ \\
  $m_i \leftarrow |C_{\ell}| $ \\
  $ T \leftarrow T \cup \{i\} $ \\
 }
 $S \leftarrow {\cal A}_{\rm RKCM}(T,z,\mathbf{w}_T=\{w_i : i \in T\},
    \mathbf{m}=\{m_i : i \in T\} )$ \\
 $r_2 \leftarrow \min_{X \subseteq T: \sum_{i \in X} m_i \geq t}
\max_{i \in X} d(i,S)$  \\
 \Return{$(S, r_{1}, r_{2})$ }
\caption{\sc RKnapCenter($\epsilon, \alpha$)}\label{alg:knap_one}
\end{algorithm}

The following two theorems bound, respectively, the
maximum value of $\tau$ set by the do-while
loop in {\sc RKnapCenter} (hence, the size of
the coreset from which the final solution is extracted), and
the approximation ratio featured by the algorithm.

\begin{theorem} \label{thm:RKCcorsize}
Assume that a $\beta$-approximation algorithm for $\tau$-center is
used in Line 8 of {\sc RKnapCenter$(\epsilon,\alpha)$} and let
$\tau_f$ be the value of $\tau$ at which the algorithm stops.  If $V$
has doubling dimension $D$, then $\tau_f = O((k+z) (c/\epsilon)^D)$,
where $k$ is the minimum cardinality of an optimal solution to the RKC
instance $(V,z,\mathbf{w})$ and $c=\beta(4\alpha+2)(\alpha+\epsilon)$.
\end{theorem}
\begin{proof}
Let $S^*_V$ be the optimal solution to the RKC instance
$(V,z,\mathbf{w})$ of minimum cardinality $k$, of cost
$r^*(V,z,\mathbf{w})$.  By reasoning as in the proof
of Lemma~\ref{lem:c1c2}, we conclude that the points of $S^*_V$
together with the at most $z$ outliers form a solution to
$(k+z)$-center on $V$ of cost at most $r^*(V,z,\mathbf{w})$. Hence,
letting $\rho^*(V, k + z)$ denote the cost of an optimal solution to
$(k + z)$-center on $V$, we have $\rho^*(V, k + z) \leq
r^*(V,z,\mathbf{w})$. Proposition~\ref{prop:dublingdim} implies that
for every $\tau \geq (c/\epsilon)^D(k+z)$, the cost of the optimal
solution to $\tau$-center on $V$ is $\rho^*(V, \tau) \leq
(\epsilon/c)\rho^*(V, k + z)$.  Let $\tau_f$ be the smallest value of
$\tau$ tested by the algorithm such that $\tau_f \geq
(c/\epsilon)^D(k+z)$ (hence, $\tau_f \leq 2(c/\epsilon)^D(k+z)$) and
let $(S,r_1,r_2)$ be the triplet returned by {\sc
  CoresetComputeAndTest} $(\tau_f)$. Observe that $r_1 \leq \beta
\rho^*(V, \tau_f)$. We now show that $r_1$ and $r_2$ satisfy the
stopping condition, thus proving the theorem. By
Point~\ref{knap_one:enum1} of Lemma \ref{lemma:knap_one}, $S$ is a
feasible solution to the RKC instance $(V, z, \mathbf{w})$ of cost at
most $2r_1 + r_2$, hence, combining this fact with the previous
observations, we have $2r_1 + r_2 \geq r^*(V,z,\mathbf{w}) \geq
\rho^*(V, k + z)$. This implies that $r_1 \leq \beta \rho^*(V, \tau_f)
\leq \beta (\epsilon/c)\rho^*(V, k + z) \leq \beta (\epsilon/c)
(2r_1+r_2)$.  By substituting $c=\beta(4\alpha+2)(\alpha+\epsilon)$
and applying trivial algebra, we obtain $\alpha(4\alpha+2)r_1 \leq
\epsilon(r_2-4\alpha r_1)$, which proves that the stopping condition
is met.
\end{proof} 

\begin{theorem}\label{thm:knap_two} 
Let $\epsilon \in (0,1)$ and let $\alpha$ be the approximation factor
of the sequential algorithm ${\cal A}_{\rm RKCM}$ for RKCM used in Line 16 
of {\sc RKnapCenter}$(\epsilon, \alpha)$. Then, the algorithm
returns an $(\alpha+\epsilon)$-approximate solution $S$ to the RKC
instance $(V,z,\mathbf{w})$.
\end{theorem}

\begin{proof}
When the
algorithm terminates returning a solution $S$, it holds that $\alpha (4\alpha+2)
r_1 \leq \epsilon (r_2-4\alpha r_1)$.  By Point~\ref{knap_one:enum1}
of Lemma \ref{lemma:knap_one}, $S$ is a solution to the RKC instance
$(V, z, \mathbf{w})$ of cost at most $2 r_{1} + r_2= r_{2}- 4\alpha
r_{1} + (4\alpha+2) r_1 \leq (r_2 - 4\alpha r_1) + (\epsilon/\alpha)
(r_2 - 4\alpha r_1)$. By Point~\ref{knap_one:enum2} of Lemma
\ref{lemma:knap_one}, $r_{2} - 4\alpha r_{1} \leq \alpha
r^*(V,z,\mathbf{w}) $. Hence, the solution returned has cost $2 r_{1}+
r_{2} \leq \alpha r^*(V,z,\mathbf{w}) + (\epsilon/\alpha)\alpha
r^*(V,z,\mathbf{w}) \leq (\alpha+\epsilon) r^*(V,z,\mathbf{w}) $.
\end{proof}
The following corollary is an immediate consequence of
Theorems~\ref{thm:RKCcorsize}, \ref{thm:knap_two} and \ref{thm:approxmult}.
\begin{corollary}\label{cor:knap}
For any fixed $\epsilon \in (0,1)$, 
the coreset-based strategy for the RKC problem presented above
can be used to compute 
a $(3+\epsilon)$-approximate solution to any instance $(M=(V,I),z)$.
If $V$ has constant doubling dimension, the  sequential running time is 
$O(|V| \mbox{\rm poly}(k,z))$.
\end{corollary}

\noindent
{\bf Remarks on the results of Sections~\ref{sec:mat} and
  \ref{sec:knap}.}  While the analysis of our algorithms is performed
in terms of the doubling dimension $D$ of $V$, the algorithms themselves are
oblivious to the value $D$. Also, it is immediate to observe that for
fixed values of $k$, $z$ and $D$, exhaustive search on the coresets
yields $(1+\epsilon)$-approximate solutions to RMC and RKC with work
merely linear in $V$.

\section{Big Data implementations}\label{sec:implementation}
In this section, we demonstrate how the RMC and RKC coreset-based
strategies presented in Sections~\ref{sec:mat} and
\ref{sec:knap} can be efficiently implemented in the MapReduce
(Subsection~\ref{subsec:MR}) and Streaming
(Subsection~\ref{subsec:STR}) settings. We will refer to a generic
instance $(M=(V,I),z)$ of RMC, with matroid $M$ of rank $k$, and a
generic instance $(V,z,\mathbf{w})$ of RKC, whose optimal solution of
minimum cardinality has size $k$. $D$ will denote the
doubling dimension of $V$.

\subsection{MapReduce implementations} \label{subsec:MR}
A MapReduce (MR) algorithm executes as a sequence of \emph{rounds}
where, in a round, a multiset of key-value pairs is transformed into a
new multiset of pairs by applying a given reduce function, referred to
as \emph{reducer}, independently to each subset of pairs having the
same key.  The model is parameterized by the total aggregate memory
available to the computation, denoted with ${\cal M}_A$, and the
maximum amount of memory locally available to each reducer, denoted
with ${\cal M}_L$. The typical goal for a MR algorithm is to run in
as few rounds as possible while keeping ${\cal M}_A$ 
(resp., ${\cal M}_L$) linear
(resp., substantially sublinear) in the input size
\cite{Dean2004,PietracaprinaPRSU12}.

A key feature of our coreset constructions for both the RMC and RKC
problems is their \emph{composability} \cite{IndykMMM14}, namely the
fact that they can be applied in parallel to subsets of an arbitrary
partition of $V$ in one MapReduce round.  Then, in a subsequent round
a solution can be computed sequentially from the coreset $T$ obtained
as the union of the coresets extracted from each subset of the
partition, using a single reducer which can fit $T$, whose size is
much smaller than $|V|$, in its local memory.

\noindent {\bf RMC problem.}  A coreset $T$ satisfying Conditions C1
and C2 of Lemma~\ref{lem:mat_one} can be constructed in one MapReduce
round as follows. Partition $V$ evenly but arbitrarily into $\ell$
disjoint subsets $V_1, \dots, V_\ell$, and assign each $V_q$ to a
distinct reducer, which builds a coreset $T^{(q)}$ for $V_q$ using the
construction described in Section~\ref{sec:mat} instantiated with the
$(\beta=2)$-approximation algorithm by Gonzalez \cite{Gonzalez85} to
find the first $k+z$ centers.  A straightforward adaptation of the
proof of Lemma~\ref{lem:c1c2} shows that $T = \cup_{1 \leq q \leq
  \ell} T^{(q)}$ satisfies conditions C1 and C2.  Setting $\ell =
\sqrt{|V|/(k(k+z))}$ and applying Theorem~\ref{thm:RMCcorsize},
we have that
$|T| = O(\sqrt{|V| k(k+z)}(8/\epsilon')^D)$. Observe that for a large
range of values of $k$ and $z$, the size of each $V_q$ and the size of
$T$ are substantially sublinear in $|V|$.  The following theorem is an
immediate consequence of the above discussion and of the results of
Section~\ref{sec:mat}.
\begin{theorem} \label{impl:MRmat}
Let $\epsilon \in (0,1)$.
There exists a 2-round MapReduce algorithm that for the RMC instance
$(M=(V,I),z)$ computes a  $(3+\epsilon)$-approximate solution using
${\cal M}_A = O(|V|)$ and ${\cal M}_L=O(\sqrt{|V| k(k+z)}(56/\epsilon)^D)$.
\end{theorem}
\noindent {\bf RKC problem.}
By reasoning as for the RMC
problem, one can easily show that each call to {\sc
  CoresetComputeAndTest}$(\tau)$ in Algorithm {\sc
  RKnapCenter}$(\epsilon,\alpha=3)$ can be implemented in 2 MapReduce
rounds with linear aggregate memory and local memory
$O(\sqrt{|V|\tau})$. Thus, by the bound on the final value
of $\tau$ proved in Theorem~\ref{thm:RKCcorsize}, Algorithm {\sc
  RKnapCenter} requires $O(\log(k+z)+D\log(1/\epsilon))$ MapReduce rounds
overall. Since, the round
complexity is a key performance indicator for efficiency in MapReduce,
we can reduce the number of rounds to $O(1/\eta)$ for any
$\eta \in (0,1)$ by substituting the progression $\tau \leftarrow 2
\tau$ (Line 5 of {\sc RKnapCenter}) with $\tau \leftarrow |V|^{\eta}
\tau$, at the expense of an extra factor $|V|^{\eta/2}$ in the local
memory requirement.  The following theorem is an immediate consequence
of the above discussion and of the results of Section~\ref{sec:knap}.

\begin{theorem} \label{impl:MRknap}
Let $\epsilon \in (0,1)$.  There exists a MapReduce algorithm that for
the RKC instance $(V,z,\mathbf{w})$ computes a
$(3+\epsilon)$-approximate solution with ${\cal M}_A = O(|V|)$, 
using
either $O(\log(k+z)+D\log(1/\epsilon))$ rounds and local memory ${\cal
  M}_L=O(\sqrt{|V| (k+z)}(c/\epsilon)^D)$, with $c \in O(1)$, or
$O(1/\eta)$ rounds with an extra factor $O(|V|^{\eta/2})$ in the
local memory size, for any $\eta \in (0,1)$.
\end{theorem}

\subsection{Streaming implementations} \label{subsec:STR}
In the streaming setting \cite{HenzingerRR98} the computation is
performed by a single processor with a small-size working memory, and
the input is provided as a continuous stream of items which is usually
too large to fit in the working memory. Typically, streaming
strategies aim at a single pass on the input but in some cases few
additional passes may be needed.  Key performance indicators are the
size of the working memory and the number of passes.

\noindent {\bf RMC problem.}  Our streaming implementation of the
coreset construction devised in Section~\ref{sec:mat} combines the
scaling algorithm of \cite{McCutchen08} with ideas introduced in
\cite{Kale19,CeccarelloPP20}. For ease of presentation, we first
describe a 2-pass implementation, and will then argue how the two
passes can be merged into a single one.  The first pass uses the
scaling algorithm to determine the set $T_{k+z}$ of $k+z$ centers
prescribed by the construction, and an upper bound $r'$ to
$r_{T_{k+z}} = \max_{j \in V}d(j,T_{k+z})$. By the analysis in
\cite{McCutchen08}, we know that $r' \leq \beta \rho^*(V,k+z)$, where
$\beta=2+\delta$ for some (arbitrarily fixed) $\delta \in (0,1)$, and
the working memory required by the computation is $O((k+z)(1/\delta)
\log (1/\delta))$. In the second pass, we incrementally build
$T_{\tau}$ using Hochbaum and Shmoys' strategy (which is naturally
streamlined) w.r.t.\ $(\epsilon'/2\beta)r'$. Concurrently, for each $i
\in T_{\tau}$ we greedily maintain the number $m_i$ of all the
points closest to $i$ seen so far and a maximal independent set of
these points.  The final coreset $T$ is the union of the
$\tau=|T_{\tau}|$ resulting independent sets. A similar strategy was
employed in \cite{Kale19}.  At the end of the pass the final solution
is computed by running the sequential 3-approximation algorithm
claimed in Theorem~\ref{thm:approxmult} on the RMCM instance
$(M_T=(T,I_T),z,\mathbf{m}=\{m_i : i \in T\})$, where $M_T$ is the
restriction of matroid $M$ to the points of $T$.  It is immediate to
see that $T$ is such that Theorems~\ref{thm:RMCcorsize}
and~\ref{thm:mat_two} hold.  The working memory required by the second
pass and by the whole algorithm is $O(|T|) =
O(k(k+z)(4\beta/\epsilon')^D)$.

The two phases described above can be merged into one by computing
$T_{\tau}$, together with the multiplicities and the related
independent sets, concurrently with the construction of $T_{k+z}$. (A
similar idea was recently employed for diversity maximization in
\cite{CeccarelloPP20}.) Specifically, the scaling algorithm maintains
$O((1/\delta) \log (1/\delta) )$ estimates of the value $r'$, with
respect to the points seen so far, and returns the tightest such
estimate at the end of the pass.  When a point $j \in V$ arrives, it
is processed according to the scaling algorithm and, concurrently, it
as also processed as prescribed by the second pass, for all available
estimates of $r'$. At the end, the algorithm returns the coreset
computed according to the best estimate. An easy argument shows that,
in this fashion, we obtain the same guarantees of the 2-pass algorithm
at the expense of a $O((1/\delta) \log (1/\delta) )$ blow-up in the
working memory space.  The following theorem is an immediate
consequence of the above discussion and of the results of
Section~\ref{sec:knap}.

\begin{theorem} \label{impl:STRmat}
Let $\epsilon, \delta \in (0,1)$.  There exists a 1-pass Streaming
algorithm that for the RMC instance $(M=(V,I),z)$ computes a
$(3+\epsilon)$-approximate solution using working memory of size
$O((1/\delta)\log (1/\delta)k(k+z)(28(2+\delta)/\epsilon)^D)$.
\end{theorem}

\noindent {\bf RKC problem.}  
A streaming implementation of Algorithm {\sc
  RKnapCenter$(\epsilon,\alpha)$} from Section~\ref{sec:knap} can be
accomplished using one pass for every iteration of the do-while loop
(Lines $2 \div 5$ of the algorithm).  Specifically, consider an
iteration of the loop for a certain value of $\tau$.  The scaling
algorithm by \cite{McCutchen08} is used to compute a
$\beta=(2+\delta)$-approximate solution $\Gamma$ to $\tau$-center on
$V$.  The algorithm returns an upper bound $r_1$ to
$\max_{j \in V}d(j,\Gamma)$ and, as noted before, $r_1
\leq \beta \rho^*(V,k+z)$.  The algorithm can also be easily
adapted to return, together with each point of $j \in \Gamma$, the
number $m_j$ of all points of $V$ at distance at most $r_1$ from $j$
which have been (implicitly) assigned to the cluster associated with
$j$, and the point $i_j$ of minimum weight among these $m_j$
points. The sequential RKCM algorithm ${\cal A}_{\rm RKCM}$ is then
run on $T =\{i_j : j \in \Gamma\}$ using $m_j$ as multiplicity of
$i_j$, for every $j$, and the original weights of the points.  The
required working memory is $O((1/\delta)\log (1/\delta) \tau)$.
Theorems~\ref{thm:RKCcorsize} and \ref{thm:knap_two}, and
Corollary~\ref{cor:knap}, immediately imply that using
$O(\log(k+z)+D\log(1/\epsilon))$ passes and the RKCM 3-approximation algorithm
claimed in Theorem~\ref{thm:approxmult} as algorithm ${\cal A}_{\rm
  RKCM}$, a $(3+\epsilon)$-approximate solution is finally computed,
and the largest working memory required by the passes is
$O((1/\delta)\log (1/\delta) (k+z)(c/\epsilon)^D)$, with
$c=(2+\delta)14(3+\epsilon)$.

In order to reduce the number of passes, we can use a coarser
progression for the values of $\tau$ by substituting $\tau \leftarrow
2 \tau$ (Line 5 of {\sc RKnapCenter}) with $\tau \leftarrow |V|^{\eta}
\tau$, for some $\eta \in (0,1)$, as was done in the
MapReduce implementation to reduce the number of rounds. In this fashion,
the number of passes shrinks to $O(1/\eta)$, at the expense of 
an extra factor $|V|^{\eta}$ in the working memory requirements.
The following theorem is an immediate consequence of the above
discussion and of the results of Section~\ref{sec:knap}.

\begin{theorem} \label{impl:STRknap}
Let $\epsilon, \delta \in (0,1)$ and $c=(2+\delta)14(3+\epsilon)$.  There exists a Streaming algorithm that
for the RKC instance $(V,z,\mathbf{w})$ computes a
$(3+\epsilon)$-approximate solution with $O(\log(k+z)+D\log(1/\epsilon))$
passes and working memory size $O((1/\delta)\log (1/\delta)
(k+z)(c/\epsilon)^D)$, or $O(1/\eta)$ passes with 
an extra factor $O(|V|^{\eta})$
in the working memory size, for any $\eta \in (0,1)$.
\end{theorem}

\section{Concluding remarks} \label{sec:conclusions}
It is not difficult to show that the techniques employed for the RMC
problem can also be used to extend the algorithms presented in
\cite{CeccarelloPP19,CeccarelloPP20} for diversity maximization under
partition and transversal matroid constraints, to work for all possible
matroids, also improving their space requirements.  The development
of 2-round/1-pass MapReduce/Streaming algorithms for the RKC problem
with low memory requirements and approximation ratios close to those
of the best sequential solutions, remain interesting open problems.


\bibliographystyle{plainurl}

\end{document}